\documentclass{article}[12pt]
\usepackage{indentfirst}
\usepackage{tabularx}
\usepackage{soul}
\usepackage[a4paper]{geometry}
\usepackage[utf8]{inputenc}
\usepackage[T1]{fontenc}
\usepackage{lmodern}
\usepackage{multirow}
\usepackage{amsmath,amssymb,amsthm} \allowdisplaybreaks[1]
\usepackage{mathrsfs}
\usepackage{graphicx}
\usepackage{bookmark}
\usepackage{booktabs}
\usepackage{float}
\usepackage{color}
\usepackage{bm}

\newtheorem{thm}{Theorem}
\newtheorem{lm}{Lemma}
\newtheorem{prop}{Proposition}

\theoremstyle{definition}
\newtheorem{de}{Definition}
\newtheorem{al}{Algorithm}

\newtheorem{re}{Remark}
\newtheorem{ex}{Example}

\newcommand{\F}{\mathbb{F}}
\newcommand{\C}{\mathcal{C}}
\renewcommand{\H}{\mathcal{H}}
\newcommand{\G}{\mathcal{G}}
\newcommand{\X}{\mathcal{X}}

\newcommand{\x}{\mathbf{x}}
\newcommand{\y}{\mathbf{y}}
\renewcommand{\u}{\mathbf{u}}
\renewcommand{\v}{\mathbf{v}}
\newcommand{\g}{\mathbf{g}}
\newcommand{\wf}{\mathbf{w}}
\newcommand{\z}{\mathbf{z}}
\newcommand{\e}{\mathbf{e}}
\newcommand{\Gf}{\mathbf{G}}
\renewcommand{\a}{\alpha}
\renewcommand{\b}{\beta}

\renewcommand{\c}{\mathbf{c}}
\newcommand{\af}{\bm{\alpha}}
\newcommand{\be}{\bm{\beta}}
\newcommand{\rf}{\bm{\gamma}}

\newcommand{\rank}{\text{rank}}
\newcommand{\CSS}{\text{CSS}}
\newcommand{\TriCode}{\text{TriCode}}
\newcommand{\rowspan}{\text{RowSpan}}
\newcommand{\modeq}[2]{#1 \equiv #2 \pmod{2}}
\newcommand{\seq}[1]{#1_1,\dots, #1_c}
\newcommand{\w}{\wedge}
\newcommand{\one}{\mathbf{1}}
\newcommand{\zero}{\mathbf{0}}

\newcommand{\red}[1]{\textcolor{black}{#1}}
\newcommand{\rR}[1]{\textcolor{black}{#1}}
\newcommand{\rrR}[1]{\textcolor{black}{#1}}

\title{Triorthogonal Codes and Self-dual Codes\thanks{This research is supported by the National Natural Science Foundation of China (12071001).}}

\author{ Minjia Shi$^\ast$\thanks{smjwcl.good@163.com},
	Haodong Lu \thanks{hdlu818@163.com}, Jon-Lark Kim\thanks{jlkim@sogang.ac.kr},
	Patrick Sol\'e\thanks{{sole@enst.fr}}
	\thanks{Minjia Shi and Haodong Lu are with the Key Laboratory of Intelligent Computing Signal
Processing, Ministry of Education, School of Mathematical Sciences, Anhui
University, Hefei 230601, China; State Key Laboratory of integrated Service Networks, Xidian University, Xi'an,
710071, China. Jon-Lark Kim is with Department of Mathematics, Sogang University, Seoul, South Korea. Patrick Sol\'e is with Aix Marseille Univ, CNRS, Centrale Marseille, I2M, Marseille, France. }}
\date{}
\begin{document}
	\maketitle
	\begin{abstract}
		Triorthogonal matrices were introduced in Quantum Information Theory in connection with distillation of magic states (Bravyi and Haah (2012)).
		We give an algorithm to construct binary triorthogonal matrices from binary self-dual codes. Further, we generalize to this setting
		the classical coding techniques of shortening and extending. We also give some simple propagation rules.
	\end{abstract}
	
	\noindent {\bf MSC (2010) :}  Primary 94B05, Secondary 81P48.\\
	{\bf Keywords:} Triorthogonal matrices, self-dual codes, propagation rules, building up
	\section{Introduction}

Quantum Information Theory has seen an exponential development since the seminal papers \cite{css2,Caldetal,css} that introduced the CSS construction, which builds quantum codes from classical codes. In these three references, self-orthogonal codes play an important role. In the last decade, a notion of a triorthogonal code, motivated by distillation of magic states
\cite{basis paper 1,optimal asymptotic,basis paper 2}, led to the notion of triorthogonal matrices. \rrR{This notion comes from a new family of distillation protocols for the state 
\[
\left.|A\right> = \left.T|+\right>\left.\sim |0\right> +  e^{i\pi/4}\left.|1\right>
\]
with a distillation cost $O(\log^{\gamma}(1/\epsilon))$ was presented in \cite{basis paper 2}, where $\gamma = \log_2(\frac{3k+8}{k})$, $k$ is an arbitrary even integer and gate $T=\exp(-i\pi Z/8)$. }

\rrR{
In \cite{basis paper 2}, the authors succeeded in achieving their objective, which is to minimize the number of raw ancillas $\rho$ required to distill magic states $\left. |A\right>$ with a desired accuracy $\epsilon$. 
Specifically, let $\sigma$ be a state of $k$ qubits which is supposed to approximate $k$ copies of $\left.|A\right>$, and suppose such a state $\sigma$ can be prepared by a distillation protocol that takes an input $n$ copies of the raw ancilla $\rho$ and uses only Clifford operations. Then the protocol has a distillation cost $C=C(\epsilon)$ if and only if $n\le Ck$. Here, $C(\epsilon)=O(\log^{\gamma}(1/\epsilon))$ and $\gamma=\frac{\log(n/k)}{\log(d_Z)}$ when we consider concatenated distillation protocol based on a triorthogonal matrix, and $[[n,k,d_Z]]$ is called the parameters of the triorthogonal codes, which will be introduced in next section. These results illustrate that the parameters can estimate the distillation cost of a quantum code which is based on a triorthogonal matrix, and this quantum code is called a triorthogonal code.
Also, a good triorthogonal code should have a small distillation cost, i.e., it has a small $\gamma$. Precise definitions are given in the next section.}

\rrR{Based on the above background, the motivations of this paper are as follows:}
\begin{itemize}
	\item \rrR{Because of the significance of the triorthogonal codes, an important and necessary problem is to construct more triorthogonal codes.}
	\item \rrR{Although the triorthogonal codes with small parameters have been classified in \cite{data paper}, but there is not a good method to flexibly combine these codes to construct new triorthogonal codes.}
\end{itemize}
\rrR{
  Therefore, the main objective of this paper is to construct triorthogonal matrices and their corresponding triorthogonal codes with arbitrary parameters.} In this paper, we study triorthogonal matrices constructed from binary self-dual codes, a special class of self-orthogonal codes. For general information on self-dual codes we refer to the classical treatise \cite{book}. Further, we give some constructions of new triorthogonal matrices from previously known ones. In particular, the classical coding techniques of shortening, extending, and propagation rules are adapted to this new situation. We also discuss the parameters of the triorthogonal codes corresponding to some of these new matrices. \rrR{With the help of these methods, many triorthogonal matrices will be discovered, and many triorthogonal codes with nice parameters will be found. These techniques have played an important role in classical error correcting codes, and we believe that their generalization in this paper will facilitate the search of triorthogonal codes.}

The material is arranged as follows. The next section contains some basic definitions and results needed for the other sections. Section 3 presents some methods for constructing new triorthogonal matrices from known ones, and gives the parameters of triorthogonal codes corresponding to some of these new matrices. Section 4 considers the influence of self-dual codes on the parameters of triorthogonal codes, and presents an algorithm to find the triorthogonal subspaces of the largest dimension of \red{a certain self-dual code}. Section 5 lists some applications and examples, and Section 6 is the conclusion of the paper.
	
	\section{Preliminaries}
	\subsection{Classical error correcting codes}
	A binary linear code is a subspace of $\F_2^c$. If a binary linear code $\C$ has dimension $r$ and length $c$, then $\C$ is a binary linear code of parameters $[c,r]$. The vectors in a code are called codewords. Let $\x=(\seq{x})$ and $\y=(\seq{y})$ be two codewords in a binary linear code $\C$ of parameters $[c,r]$, we define two operations between $\x$ and $\y$ as follows:
	\[\begin{aligned}
		\x + \y = (x_1+y_1,\dots,x_c+y_c),\;
		\x \w \y = (x_1y_1,\dots,x_cy_c).
	\end{aligned}\]
	Let $\x$ and $\y$ be in $\F_2^c$.
	The (Hamming) weight of $\x$, denoted by ${\mbox{wt}}(\x)$, is the number of nonzero coordinates of $\x$, and the (Hamming) distance between $\x$ and $\y$ is defined as $d(\x,\y)={\mbox{wt}}(\x+\y)$. \red{Let $\C$ be a $[c,r]$ binary linear code. If the minimum distance of $\C$ is
	$d = d(\C) = \min\{{\mbox{wt}}(\c):\c\in\C\setminus\{\zero_c\}\}$ where $\zero_c$ is the all-zero vector of length $c$, then we say that $\C$ has parameters $[c,r,d]$.} Correspondingly, for a set (not necessarily linear) $H\subseteq \F_2^c$, we define
	\[
	{\mbox{wt}}(H) = \min\{ {\mbox{wt}}(\x):\x\in H\text{ and }\x\neq \zero_c \}.
	\]
	
	We use the notation $|\x|$ instead of ${\mbox{wt}}(\x)$ for the convenience of description, i.e., $|\x| = {\mbox{wt}}(\x)$. It is clear that $|\x\w\y|\pmod{2}$ is the Euclidean inner product of $\x$ and $\y$ over $\F_2$. We denote $\C^{\perp}$ as the dual of a binary code $\C$ of parameters $[c,r]$, where
	\[
	\C^{\perp} = \{ \c\in\F_2^c: |\c\w\x|\equiv 0 \pmod{2} \text{ for all } \x\in\C \}.
	\]
	Then a binary code $\C$ is self-dual if $\C=\C^{\perp}$, and self-orthogonal if $\C\subseteq \C^{\perp}$. We denote $\H=\rowspan\{H\}$ as the linear span of a set $H\subseteq \F_2^c$. Here are some simple properties.
	\begin{prop}
		If $\x,\y,\z$ are vectors in $\F_2^c$, then $\x\w \x = \x$, $\x\w \one_c=\x$, $\x\wedge(\y+\z) = (\x\wedge \y) + (\x\wedge \z)$, where $\one_c$ is the all-one vector of length $c$.
	\end{prop}
	
	\begin{prop} If $\x,\y,\z$ are vectors in $\F_2^c$, then
		$|\x\w (\y+\z)| \equiv |\x\w\y|+|\x\w\z|\pmod{2}$.
	\end{prop}
	
	\begin{prop}
		If $\C_1\subseteq \C_2\subseteq \F_2^c$ are two binary linear codes, then $(\C_1^{\perp})^\perp = \C_1$, $\C_2^{\perp}\subseteq \C_1^{\perp}$ and $d(\C_1)\geq d(\C_2)$.
	\end{prop}

\subsection{Triorthogonality: Matrices and codes}
To introduce triorthogonal codes, we introduce the following concepts.
\begin{de}[Triorthogonal matrices in \cite{basis paper 1}] \label{tri matrix}
	A binary matrix $G = (G_{ij})_{m\times n}$ of size $m$-by-$n$ is called triorthogonal if \\
	(1) for all pairs $(a,b)$ that satisfy $1\leq a < b \leq m$, we have
	\[
	\modeq{\sum_{j=1}^n G_{aj}G_{bj}}{0}, {\mbox{ and}}
	\]
	(2) for all triples $(a,b,c)$ that satisfy $1\leq a < b < c \leq m$, we have
	\[
	\modeq{\sum_{j=1}^n G_{aj}G_{bj}G_{cj}}{0}.
	\]
\end{de}

Given a full-rank triorthogonal matrix $G_{m\times n}$ of size $m$-by-$n$, by row permutation, we can always divide the matrix $G$ into two parts, namely
\[G\xrightarrow{\text{Row Permutation}}
\left[ \begin{array}{c}
	G_{1}\\
	\hline
	G_{0}
\end{array} \right]
\begin{array}{l}
	\} \text{ odd weight} \\
	\} \text{ even weight}
\end{array},
\]
where all the rows with odd weights in $G$ form $G_1$, and the remaining rows form $G_0$. In this paper, we will always assume that the first $k$ rows of $G$ have odd-weight, i.e., $\sum_{j=1}^n G_{aj} \equiv 1 \pmod{2}$ for $1\leq a\leq k$ and the remaining rows have even-weight, i.e., $\sum_{j=1}^n G_{bj} \equiv 0 \pmod{2}$ for $k+1\leq b\leq m$. Now we give the concept of triorthogonal codes.

\begin{de}[Triorthogonal codes in \cite{basis paper 1}] \label{tri codes}
	Let $G$ be a binary triorthogonal matrix of size $m$-by-$n$ which has $k$ odd-weight rows. Let $\G_0$ denote the span of all the even weight rows of $G$, and $\G$ denote the span of all the rows of $G$. Then a quantum CSS code, by letting $\G_0$ correspond to $X$-stabilizers, and $\G^{\perp}$ to $Z$-stabilizers, is called the triorthogonal code.
\end{de}

	In \cite{basis paper 2}, it is  shown that such a triorthogonal code has $k$ logical qubits, i.e., encodes $k$ logical qubits into $n$ qubits. We are interested in the minimum weight $d_Z$ of any nontrivial $Z$-logical operators of the triorthogonal codes, that are related to a triorthogonal matrix $G$. We call such a number the distance of the matrix $G$, i.e.,
\[
d_Z = \min\{{\mbox{wt}}(c):c\in\G_0^\perp\setminus \G^\perp\} = {\mbox{wt}}(\G_0^\perp\setminus \G^\perp).
\]
In the rest of this paper, a triorthogonal code has parameters $[[n,k,d_Z]]$ if this code encodes $k$ logical qubits into $n$ qubits, and the distance of the corresponding triorthogonal matrix is $d_Z$. Furthermore, for the convenience of description, given a triorthogonal matrix $G$, we use the notation \red{$\TriCode(G) =
\CSS(\G_0;\G)$} to describe the triorthogonal code determined by $G$, where $\G_0$ is the span of all the even-weight rows of $G$ and $\G$ is the span of all the rows of $\G$.

\subsection{Triorthogonal spaces and triorthogonal matrices}
We introduce a class of linear subspaces of $\F_2^c$ that is closely related to triorthogonal matrices.
\begin{de}[\cite{data paper}]\label{tri space}
	A subspace $\H\subseteq \F_2^c$ is triorthogonal if for any three vectors $\x,\y,\z\in\H$, we have $\modeq{|\x\w\y\w\z|}{0}$. If $\H$ contains all-one vector $\one_c$, then $\H$ is called unital.
\end{de}
For each linear subspace, a matrix whose rows form a basis of this subspace can describe it completely. Conversely, we can use the linear span of all rows of some special matrices to describe orthogonal or triorthogonal subspace. We call such a matrix the generator matrix of the corresponding subspace. Obviously, all rows of any generator matrix of a triorthogonal subspace have even weights.

Now we would use the generator matrices of the triorthogonal subspaces to gain some triorthogonal matrices which have odd-weight rows. For a linear space $\C$ of parameters $[c,r]$ and any positive integer $k$ $(k\leq r)$, there exists a generator matrix $G$ of $\C$ such that
\[G = \left[ \begin{array}{c|c}
	I_r &P \\
\end{array}\right] =
\left[ \begin{array}{c|c|c}
	I_k &O &P_1 \\
	\hline
	O   &I_{r-k} &P_0 \\
\end{array}\right],
\]
where $I_k$ is the identity matrix of size $k$-by-$k$, and $O$ is the all-zero matrix of suitable size. Then \[ \left[ \begin{array}{c}
	G_1 \\
	\hline
	G_0
\end{array}\right] =
\left[ \begin{array}{c|c}
	O	&P_1 \\
	\hline
	I_{r-k} &P_0\\
\end{array}\right]\]
is a triorthogonal matrix, where the rows of $G_1$ have odd weights and the rows of $G_0$ have even weights.

In the following sections, we will follow this line of thought and use some classical coding theory to construct some triorthogonal codes.

\section{Triorthogonal codes from matrices}
In this section, we will construct new triorthogonal matrices from old ones, and the parameters of the triorthogonal codes corresponding to these new matrices are discussed.
\subsection{Shortening and extending}
Given a triorthogonal matrix $G$ with $k$ odd-weight rows, we would like to shorten it on one position to gain new triorthogonal matrices. For a binary vector $\c$ of length $n$ and a given $i$ $(1 \le i \le n)$, if there is a 0 in the $i$-th position of $\c$,  we denote it by $(\c |_i 0)$. For a set $S\subseteq \F_2^n$, we denote the number of elements in $S$ as $|S|$. Shortening $S$ on the $i$-th coordinate position means the set of vectors of length $n-1$ obtained by removing the $i$-th column of $(\x|_i0)\in S$, that is,
\[
S_i = \{\x:(\x|_i0)\in S\} \text{ for } 1\leq i\leq n.
\]
\begin{lm}
	Given a full-rank triorthogonal matrix $G$ of size $m$-by-$n$, we denote the $i$-th row of $G$ as $\g_i$ and $S=\{\g_1,\dots,\g_m\}$. If $S_i=\{\x:(\x|_i 0)\in S\}$, then the matrix $G_i$ whose rows are the elements in $S_i$ is triorthogonal.
\end{lm}
\begin{proof}
	For any $\x',\y'$ $(\x'\neq\y')$ in $S_i$, let $\x = (\x'|_i0)$ and $\y=(\y'|_i0)$. Then we have $\x,\y\in S$ and
	\[\begin{aligned}
		|\x\w\y| = \sum_{j=1}^n x_jy_j =  \sum_{j=1,j\neq i}^n x_jy_j = |\x'\w\y'| \equiv 0\pmod{2}.
	\end{aligned}\]
	For any distinct three elements $\x',\y',\z'$  in $S_i$, let $\x = (\x'|_i0)$, $\y=(\y'|_i0)$, $\z=(\z'|_i0)$ and $\x,\y,\z\in S$. Then we have
	\[\begin{aligned}
		|\x\w\y\w\z| = \sum_{j=1}^n x_jy_jz_j =  \sum_{j=1,j\neq i}^n x_jy_jz_j = |\x'\w\y'\w\z'| \equiv 0\pmod{2}.
	\end{aligned}\]
	Therefore, $G_i$ is triorthogonal.
\end{proof}

\begin{thm}
	Given a full-rank triorthogonal matrix $G$ of size $m$-by-$n$ with $k$ $(k>0)$ odd-weight rows, and the parameters of $\TriCode(G)$ are $[[n,k]]$. Let
	\[\begin{aligned}
		S &= \{ \g : \g\text{ is the row of }G \},\\
		S^1 &= \{\g_1: \g_1 \text{ is the odd-weight row of }G\},\text{ and }\\
		S^0 &= \{\g_0: \g_0 \text{ is the even-weight row of }G\}.\\
	\end{aligned}\]
	Choosing $1\leq i\leq n$ to satisfy $S^1_i\neq \emptyset$ and $S^0_i\neq\emptyset$ $(S^1_i\cup S^0_i = S_i)$, if all elements of $S_i$ as rows form a full-rank matrix $G_i$ of size $|S_i|$-by-$(n-1)$, we have a triorthogonal code $\TriCode(G_i)$ of parameters $[[n-1,|S_i^1|]]$. Here, $S_i^1 = \{\x:(\x|_i0)\in S^1\}$, $S_i^0 = \{\x:(\x|_i0)\in S^0\}$ and $S_i = \{\x:(\x|_i0)\in S\}$.
\end{thm}
\begin{proof}
	Since $G_i$ is a full-rank triorthogonal matrix of size $|S_i|$-by-$(n-1)$ and $G_i$ has $|S_i^1|$ odd-weight rows, we obtain this result.
\end{proof}

Let $ \e_i=(0, \dots,0,1,0, \dots,0)^T$ be the column vector of length $n$, where 1 occurs on the $i$-th position $(1 \leq i \leq n$). One can also extend a triorthogonal matrix as follows.

\begin{lm} \label{lem-ext}
	Let $G$ be a binary triorthogonal matrix of size $m$-by-$n$.  Then
	\[
	G'=\left[\begin{array}{c|c}
		\e_i  &G
	\end{array}
	\right]
	\]
	is a triorthogonal matrix of size $m$-by-$(n+1)$, \red{where $1\leq i\leq m$.}
\end{lm}

\begin{thm}\label{ext-CSS}
	Let $G$ be a full-rank binary triorthogonal matrix of size $m$-by-$n$, and the first $k$ $(k>0)$ rows of $G$ have odd weight and the remaining rows have even weight. Let $G' = \left[ \begin{array}{c|c}
		\e_i &G
	\end{array}\right]$ $(1\leq i\leq k)$. If the triorthogonal code $\TriCode(G)$ has parameters $[[n,k]]$, then $\TriCode(G')$ has parameters $[[n+1,k-1]]$.
\end{thm}

\begin{ex}\label{eg-short-2}
	\red{In \cite{basis paper 2}, the matrix
	\[
	G = \left[\begin{array}{cccccccccccccc}
		1 &1 &1 &1 &1 &1 &1 &0 &0 &0 &0 &0 &0 &0  \\
		0 &0 &0 &0 &0 &0 &0 &1 &1 &1 &1 &1 &1 &1  \\
		1 &0 &1 &0 &1 &0 &1 &1 &0 &1 &0 &1 &0 &1  \\
		0 &1 &1 &0 &0 &1 &1 &0 &1 &1 &0 &0 &1 &1  \\
		0 &0 &0 &1 &1 &1 &1 &0 &0 &0 &1 &1 &1 &1  \\
	\end{array}
	\right]_{5\times 14}
	\]
	is triorthogonal and $G$ can generate a triorthogonal code of parameters $[[14,2,2]]$.} By extending $G$ by one column, we get a triorthogonal matrix $G_3$:
	\[
	G_3 = \left[\begin{array}{c|cccccccccccccc}
		0&	1 &1 &1 &1 &1 &1 &1 &0 &0 &0 &0 &0 &0 &0  \\
		1&	0 &0 &0 &0 &0 &0 &0 &1 &1 &1 &1 &1 &1 &1  \\
		0&	1 &0 &1 &0 &1 &0 &1 &1 &0 &1 &0 &1 &0 &1  \\
		0&	0 &1 &1 &0 &0 &1 &1 &0 &1 &1 &0 &0 &1 &1  \\
		0&	0 &0 &0 &1 &1 &1 &1 &0 &0 &0 &1 &1 &1 &1  \\
	\end{array}
	\right]_{5\times 15},
	\]
	and $G_3$ can generate a triorthogonal code of parameters $[[15,1,3]]$.
\end{ex}

Here, we give a construction of the triorthogonal spaces, which is obtained by shortening the triorthogonal spaces.
\begin{lm}
	Let $\C$ be a linear code of parameters $[n,k,d]$. If $\C$ is a triorthogonal space, then $\C_i = \{\x:(\x|_i0)\in\C\}$ is also a triorthogonal space, and $\C_i$ is a linear code of parameters $[n-1,k-1,d'\geq d]$. Here $1\leq i\leq n+1$.
\end{lm}

\subsection{Propagation rules}
We use some simple propagation rules to generate more triorthogonal matrices.

\begin{prop}\label{direct sum 1}
	If $A$ and $B$ are two binary triorthogonal matrices with the same number of rows, then
	\[G=\left[ \begin{array}{c|c}
		A &B\\
	\end{array}\right] \]
	is also a triorthogonal matrix.
\end{prop}

\begin{proof}
	Let the $i$-th rows of $G,A,B$ be $\g_i,\af_i,\be_i$, respectively. For any $1\leq i<j\leq m$, where $m$ is the number of rows in $G$, we have
	\[\begin{aligned}
		|\g_i\w \g_j| &= |\af_i\w \af_j| + |\be_i\w \be_j| \equiv 0\pmod{2}. \\
	\end{aligned}\]
	For any $1\leq i<j<k\leq m$, we have
	\[
	|\g_i\w \g_j\w \g_k| = |\af_i\w \af_j \w \af_k| + |\be_i\w \be_j\w \be_k| \equiv 0\pmod{2}.
	\]
	Therefore, $G$ is triorthogonal.
\end{proof}
\begin{lm} \label{gene matrix of direct sum}
	If $A$ and $B$ are two triorthogonal matrices, then
	\[G=\left[ \begin{array}{c|c}
		A &O\\
		\hline
		O &B\\
	\end{array}\right] \]
	is also a triorthogonal matrix. Here, $O$ is the all-zero matrix.
\end{lm}
\begin{proof}
	It is clear that
	\[
	\left[ \begin{array}{c}
		A\\
		O\\
	\end{array}\right]
	\text{ and }
	\left[ \begin{array}{c}
		O\\
		B\\
	\end{array}\right]
	\] are all triorthogonal matrices, and we obtain the result by Proposition \ref{direct sum 1}.
\end{proof}

\begin{lm}[\cite{book}] \label{direct sum}
	Let $\C_1$ and $\C_2$ be two binary linear codes with generator matrices $A_{a_1\times a_2}$ and $B_{b_1\times b_2}$, respectively. Let
	\[
	\C = \C_1\oplus\C_2 = \{(\u|\v):\u\in\C_1,\v\in\C_2\},
	\]
	where $(\u|\v) = (u_1,\dots,u_{a_2},v_1,\dots,v_{b_2})$. We have\\
	(1) $\C$ is a binary $[a_2+b_2,a_1+b_1,d]$ linear code, where $d = \min\{d(\C_1),d(\C_2)\}$;\\
	(2) a generator matrix of $\C$ can be
	\[\left[ \begin{array}{c|c}
		A &O \\
		\hline
		O &B
	\end{array}\right] .
	\]
\end{lm}

\begin{lm}\label{oplus and perp}
	Following the notations of Lemma \ref{direct sum}, we have	$\C_1^{\perp}\oplus\C_2^{\perp} = (\C_1\oplus\C_2)^{\perp}$.
\end{lm}
\begin{proof}
	For any $(\u|\v) \in \C_1^{\perp}\oplus\C_2^{\perp}$ and $(\x|\y)\in \C_1\oplus\C_2$, we have
	\[
	|(\u|\v)\w (\x|\y)| = |\u\w\x| + |\v\w\y| \equiv 0 \pmod{2}.
	\]
	Therefore, $(\u|\v)\in(\C_1\oplus\C_2)^{\perp}$ and $\C_1^{\perp}\oplus\C_2^{\perp} \subseteq (\C_1\oplus\C_2)^{\perp}$. Noting that the dimension of $\C_1^{\perp}\oplus\C_2^{\perp}$ is $(a_2-a_1)+(b_2-b_1)$, and the dimension of $(\C_1\oplus\C_2)^{\perp}$ is $(a_2+b_2) - (a_1+b_1)$. Hence, $|\C_1|=|\C_2|$ and $\C_1^{\perp}\oplus\C_2^{\perp} = (\C_1\oplus\C_2)^{\perp}$.
\end{proof}

\begin{lm}\label{diff oplus}
	\red{If $A,B,C,D$ are four binary linear codes}, where $C\subsetneqq A$ and $D\subsetneqq B$, then
	\[
	(A\oplus B)\setminus(C\oplus D) = [(A\setminus C)\oplus(B\setminus D)] \cup [C\oplus(B\setminus D)] \cup [(A\setminus C)\oplus D].
	\]
	Furthermore, ${\rm{wt}}((A\oplus B)\setminus(C\oplus D)) = \min \{{\rm{wt}}(A\setminus C),{\rm{wt}}(B\setminus D)\}$.
\end{lm}
\begin{proof}
	It is clear that $(A\setminus C)\oplus(B\setminus D)$, $C\oplus(B\setminus D)$, $(A\setminus C)\oplus D$ are  subsets of $(A\oplus B)\setminus(C\oplus D)$. Therefore, $[(A\setminus C)\oplus(B\setminus D)] \cup [C\oplus(B\setminus D)] \cup [(A\setminus C)\oplus D]\subseteq (A\oplus B)\setminus(C\oplus D)$. Also,
	\[\begin{array}{rlllc}
		(A\setminus C)\oplus(B\setminus D) &\cap  &C\oplus(B\setminus D) &= &\varnothing, \\
		C\oplus(B\setminus D) &\cap &(A\setminus C)\oplus D &= &\varnothing, \\
		(A\setminus C)\oplus(B\setminus D)&\cap &(A\setminus C)\oplus D&= &\varnothing.
	\end{array}\]
	Since
	\[\begin{aligned}
		|(A\oplus B)\setminus(C\oplus D)| &= |A||B|-|C||D|,     \\
		|(A\setminus C)\oplus(B\setminus D)| &= (|A|-|C|)(|B|-|D|)   =   |A||B| + |C||D| - |A||D| - |B||C|,  \\
		|C\oplus(B\setminus D)| &=  |C|(|B|-|D|)  = |B||C| - |C||D|, \\
		|(A\setminus C)\oplus D| &= (|A|-|C|)|D| =  |A||D| - |C||D|.
	\end{aligned}\]
	Therefore, $|(A\setminus C)\oplus(B\setminus D)|+|C\oplus(B\setminus D)|+|(A\setminus C)\oplus D| = |(A\oplus B)\setminus(C\oplus D)|$. The result follows.
	
	For the remaining conclusion, we divide it into two parts to prove. \red{If ${\mbox{wt}}((A\setminus C)\oplus(B\setminus D)) \neq {\mbox{wt}}(A\setminus C)+{\mbox{wt}}(B\setminus D)$, then there exists $(\x|\y)\in (A\setminus C)\oplus(B\setminus D)$ satisfying \[{\mbox{wt}}(\x)+{\mbox{wt}}(\y)={\mbox{wt}}((\x|\y))< {\mbox{wt}}(A\setminus C)+{\mbox{wt}}(B\setminus D),\]
	where $\x\in A\setminus C$ and $\y\in B\setminus D$. Therefore, we have ${\mbox{wt}}(\x)<{\mbox{wt}}(A\setminus C)$ or ${\mbox{wt}}(\y)<{\mbox{wt}}(B\setminus D)$, or both, which is a contradiction.} Hence,
	\[\begin{aligned}
		{\mbox{wt}}((A\oplus B)\setminus(C\oplus D)) &= \min\{ {\mbox{wt}}((A\setminus C)\oplus(B\setminus D)) , {\mbox{wt}}(C\oplus(B\setminus D)), {\mbox{wt}}((A\setminus C)\oplus D)\}\\
		&= \min\{ {\mbox{wt}}(A\setminus C)+{\mbox{wt}}(B\setminus D), {\mbox{wt}}(A\setminus C), {\mbox{wt}}(B\setminus D) \}\\
		&= \min\{{\mbox{wt}}(A\setminus C), {\mbox{wt}}(B\setminus D)\}.
	\end{aligned}\]
	The result follows
\end{proof}

\begin{thm}\label{dirct sum of tricodes}
	Let $A_{\a_1\times\a_2}$ and $B_{\b_1\times\b_2}$ be two triorthogonal matrices of distance $d_A$ and $d_B$, respectively. The first $k_a$ $(k_a>0)$ rows of $A$ have odd weight and the remaining rows have even weight. The first $k_b$ $(k_b>0)$ rows of $B$ have odd weight and the remaining rows have even weight. Let
	\[
	A = \left[\begin{array}{c}
		A_1 \\
		\hline
		A_0
	\end{array} \right]\text{ and }
	B = \left[\begin{array}{c}
		B_1 \\
		\hline
		B_0
	\end{array}\right],
	\] where $A_1$ has size $k_a$-by-$\a_2$ and $B_1$ has size $k_b$-by-$\b_2$. If
	\[
	G = \left[\begin{array}{c|c}
		A_1 &O \\
		\hline
		O   &B_1 \\
		\hline
		A_0 &O \\
		\hline
		O   &B_0 \\
	\end{array} \right] \xleftarrow{\text{Row Permutation}}
	\left[\begin{array}{c|c}
		A &O\\
		\hline
		O &B\\
	\end{array} \right],
	\]
	then the parameters of $\TriCode(G)$ are $[[\a_2+\b_2,k_a+k_b,\min\{d_a,d_b\}]]$, where the parameters of $\TriCode(A)$ and $\TriCode(B)$ are $[[\a_2,k_a,d_a]]$ and $[[\b_2,k_b,d_b]]$, respectively.
\end{thm}
\begin{proof}
	Let $\TriCode(G) = \CSS(\G_0;\G)$, $\TriCode(A) = \CSS(\G_{a0};\G_a)$ and $\TriCode(B) = \CSS(\G_{b0};\G_b)$. Since $G$ has $k_a+k_b$ odd-weight rows, thus $\TriCode(G)$ has parameters $[[\a_2+\b_2,k_a+k_b]]$. Since
	$
	\G = \G_a \oplus \G_b,
	$
	and $\G_0 = \G_{a0}\oplus\G_{b0}$, then by Lemma \ref{oplus and perp} and Lemma \ref{diff oplus}, we have
	\[\begin{aligned}
		{\mbox{wt}}(\G_0^{\perp}\setminus \G^{\perp}) &= {\mbox{wt}}((\G_{a0}\oplus\G_{b0})^{\perp}\setminus (\G_a \oplus \G_b)^{\perp}) \\
		&= {\mbox{wt}}((\G_{a0}^{\perp}\oplus\G_{b0}^{\perp})\setminus (\G_a^{\perp} \oplus \G_b^{\perp}) \\
		&= \min\{
		{\mbox{wt}}(\G_{a0}^\perp\setminus\G_{a}^\perp),{\mbox{wt}}(\G_{b0}^\perp\setminus\G_{b}^\perp)
		\}\\
		&= \min\{  d_a , d_b  \}.
	\end{aligned}\]
	The result follows.
\end{proof}
\begin{re}
	Theorem \ref{dirct sum of tricodes} shows that if there exist two triorthogonal codes of parameters $[[n_i,k_i,d_i]]$ $(i=1,2)$, then there exists a triorthogonal code of parameters $[[n_1+n_2,k_1+k_2,\min\{d_1,d_2\}]]$. Here $d_i$ $(i=1,2)$ are the distances of the corresponding triorthogonal matrices.
\end{re}

Next, for two binary linear codes $\C_1$ and $\C_2$ of the same length, we define their Plotkin sum as follows.
\[
\C_1 \odot \C_2 = \{ (\u|\u+\v):\u\in\C_1,\v\in\C_2 \}.
\]
\begin{lm}[\cite{book}] \label{u u+v}
	Let $\C_1$ and $\C_2$ be two binary linear codes of parameters $[n,k_1,d_1]$ and $[n,k_2,d_2]$, respectively. Let
	\[\C = \C_1 \odot \C_2 = \{ (\u|\u+\v):\u\in\C_1,\v\in\C_2 \}.\]
	Then\\
	(1) $\C$ is a binary $[2n,k_1+k_2,d]$ linear code, where $d = \min\{2d_1,d_2\}$; \\
	(2) a generator matrix of $\C$ can be
	\[
	G = \left[ \begin{array}{c|c}
		G_1  &G_1 \\
		\hline
		O    &G_2 \\
	\end{array} \right],
	\]
	where $G_i$ is the generator matrix of $\C_i$, $i=1,2$.
\end{lm}

Unfortunately, if $G_1$ and $G_2$ are two triorthogonal matrices, then $G$ may not be a triorthogonal matrix. But we can consider the following two constructions. Let $G_{m\times n}$ be a triorthogonal matrix whose first $k$ rows have odd weight. If
\[
G = \left[\begin{array}{c}
	G_1\\
	\hline
	G_0\\
\end{array}\right]
\begin{array}{l}
	\}\text{first }k \text{ rows } \\
	\}\text{remaining } m-k \text{ rows }
\end{array},
\]
then the matrices
\[	G'=\left[ \begin{array}{c|c}
	G_1  &G_1 \\
	\hline
	O    &G_0 \\
\end{array}\right]
\text{ and }
G''=\left[ \begin{array}{c|c}
	O  &G_1 \\
	\hline
	G_0    &G_0 \\
\end{array}\right]
\]
are all triorthogonal by Proposition \ref{direct sum 1}. $G'$ can generate a triorthogonal space since each row of $G'$ has even weight, and $G''$ can generate a triorthogonal code of parameters $[[2n,k]]$. Furthermore, we have the following theorem.

\begin{thm}
	\red{Over $\F_2$, let $G_s$} be a binary triorthogonal matrix of size $m$-by-$n$ whose first $k$ rows have odd weight. If
	\[
	G_s = \left[\begin{array}{c}
		G_1\\
		\hline
		G_0\\
	\end{array}\right]
	\begin{array}{l}
		\}\text{first }k \text{ rows } \\
		\}\text{remaining } m-k \text{ rows }
	\end{array}
	\text{ and }
	G=\left[ \begin{array}{c|c}
		O  &G_1 \\
		\hline
		G_0    &G_0 \\
	\end{array}\right],
	\]
	then the triorthogonal code $\TriCode(G)$ has parameters $[[2n,k,d_Z]]$, where $d_Z \geq d((\G_0^2)^\perp)$. Here $\G_1$ is the span of all the rows of $G_1$ and $\G_0^2$ is the span of all the rows of $[G_0\mid G_0]$.
\end{thm}

\subsection{New triorthogonal matrices from known ones}
At the end of this section, we give some useful ways to construct new triorthogonal matrices \red{from} known ones.
\begin{prop}\label{direct sum 2}
	Let $G_0$ be a binary triorthogonal matrix of size $m$-by-$n$. If a matrix $A$ is made up of any $r$ $(2r\leq m)$ rows of $G_0$, and $r$ rows of the remaining $m-r$ rows of $G_0$ form a matrix $B$, then $A+B$ is triorthogonal.
\end{prop}
\begin{proof}
	Let the $i$-th rows of $A$ and $B$ be $\af_i$ and $\be_i$, respectively. Then for any $1\leq i<j<k\leq r$, we have
	\[\begin{aligned}
		&|(\af_i+\be_i)\w(\af_j+\be_j)| \\
		&\equiv |\af_i\w \be_j| + |\be_i\w \af_j| + |\af_i\w \be_j| + |\be_i\w \be_j| \pmod{2}\\
		&\equiv 0\pmod{2},\text{ and }\\
		&|(\af_i+\be_i)\w(\af_j+\be_j)\w(\af_k+\be_k)|  \\
		&\equiv |\af_i\w \af_j\w \af_k| + |\af_i\w \af_j\w \be_k| + |\af_i\w \be_j\w \af_k| + |\af_i\w \be_j\w \be_k| \\
		&+ |\be_i\w \af_j\w \af_k| + |\be_i\w \af_j\w \be_k| + |\be_i\w \be_j\w \af_k|+ |\be_i\w \be_j\w \be_k| \pmod{2} \\
		&\equiv 0\pmod{2}.
	\end{aligned}\]
	These show that $A+B$ is triorthogonal.
\end{proof}

In what follows, we describe a building-up type construction~\cite{build up over F2} to obtain a triorthogonal matrix of a larger size from a given triorthogonal matrix.

\begin{thm}\label{new building up}
	Let \[G_0 = \begin{bmatrix}
		&	\g_1  & \\
		&	\g_2  & \\
		&	\vdots & \\
		&	\g_n & \\
	\end{bmatrix}\] be a triorthogonal matrix of size $n$-by-$m$. Let $\x$ be a vector of length $m$ over $\F_2$, $\one$ and $\bm{0}$ be the all-one and the all-zero vectors of length $m$, respectively. We have
	\[
	G =\left[  \begin{array}{c|c|cc}
		\one &\bm{0} &\x\\
		\hline
		\y_1 &\y_1 &\g_1\\
		\vdots &\vdots &\vdots\\
		\y_n &\y_n &\g_n \\
	\end{array} \right]
	\] is a triorthogonal matrix, where $\y_i = \x\w \g_i$.
\end{thm}
\begin{proof}
	Let $\Gf_i$ be the $i$-th row of $G$. First we check the rows except the first row. For any $1\leq i<j<k\leq n$, we have
	\[
	\begin{aligned}
		|\Gf_{i+1}\w \Gf_{j+1}| &= 2|\y_i\w \y_j| + |\g_i\w \g_j| \equiv 0\pmod{2}, \\
		|\Gf_{i+1}\w \Gf_{j+1}\w \Gf_{k+1}| &= 2|\y_i\w \y_j\w \y_k| + |\g_i\w \g_j\w \g_k| \equiv 0\pmod{2}.
	\end{aligned}
	\]
	Now we check the first row. Let $\af = (\bm{1},\bm{0},\x)$. For any $1\leq i<j\leq n$, we have
	\[
	\begin{aligned}
		|\af \w \Gf_{i+1}| &= |\one\w\y_i| + |\x\w \g_i| \\
		&= 2|\x\w \g_i|\equiv0\pmod{2}, \\
		|\af\w \Gf_{i+1}\w \Gf_{j+1}| &= |\one\w \y_i\w \y_j| + |\x\w \g_i\w \g_j|\\
		&= |(\x\w \g_i) \w (\x\w \g_j)| + |\x\w \g_i\w \g_j|\\
		&= 2|\x\w \g_i \w \g_j| \equiv 0\pmod{2}. \\
	\end{aligned}
	\]By the definition of the triorthogonal matrices, we obtain the conclusion.
\end{proof}

\section{Triorthogonal codes from self-dual codes}
In this section, we first discuss a class of triorthogonal matrices derived from generator matrices of self-dual codes. Considering the speciality of self-dual codes, we will give an algorithm to find as many triorthogonal spaces as possible from self-dual codes. In this section, we only consider full-rank matrices.

\subsection{Background material}
An immediate question is when the generator matrix of a self-dual code is triorthogonal?
A complete characterization of generator matrices of self-dual codes that are triorthogonal is shown in Theorem \ref{self-triorthogonal}. We first give a useful necessary and sufficient condition.

\begin{thm}\label{judge}
	If $G$ is a generator matrix of a self-dual code $\C$, then $G$ is triorthogonal if and only if $\x\w \y\in\C$ for any $\x,\y\in\C$.
\end{thm}

\begin{proof}
	Suppose that $G$ is triorthogonal. Because $G$ is also the generator matrix of a self-dual code $\C$, each row of $G$ has even Hamming weight and thus $\C$ is a triorthogonal space, which means $|\x\w \y\w \z| \equiv 0\pmod{2}$ for any $\x,\y,\z\in\C$. Fixing $\x,\y$, we have $|\x\w \y\w \z|\equiv0\pmod{2}$ for all $\z\in\C$, so $\x\w \y\in\C^{\perp} =\C$. Then by the arbitrariness of $\x, \y$, we get the result we need.
	
	On the other hand, if we have $\x\w \y\in \C$ for any $\x,\y\in\C$, then we have
	\[|\x\w \y\w \z| = |(\x\w \y)\w \z| \equiv 0\pmod{2} {\mbox{ for any }} \x,\y,\z\in\C \]
	since $\C$ is an orthogonal space. Therefore, $\C$ is a triorthogonal space and its generator matrix $G$ is triorthogonal. The result follows.
\end{proof}

One can just check the basis of a linear code $\C$ to see if $\C$ is a triorthogonal space.

\begin{lm}\label{tri basis is tri}
	Let $H$ be a basis of a binary linear code $\C$. If $\af\w\be\in\C$ for any $\af,\be\in H$, then we have $\x\w \y\in\C$ for any $\x,\y\in\C$.
\end{lm}

\begin{proof}
	Let $H = \{\seq{\af}\}$ and $\x = \sum_{i\in\Lambda_x}\af_i,\;\y = \sum_{j\in\Lambda_y}\af_j$ for some index sets $\Lambda_x, \Lambda_y$. Then, since $\C$ is linear, we have
	\[
	\x\w \y
	=  (\sum_{i\in\Lambda_x}\af_i)\w(\sum_{j\in\Lambda_y}\af_j)
	= \sum_{i\in\Lambda_x}\sum_{j\in\Lambda_y}(\af_i\w \af_j) \in\C
	\]because $\C$ is linear.		
\end{proof}

\begin{ex} \label{all tri}
	Let $G_1 = [\;I_k\mid I_k\;]$, where $I_k$ is the identity matrix of size $k$-by-$k$. Then $G_1$ is a triorthogonal matrix and $\G_1=\rowspan\{G_1\}$ is a triorthogonal space, since any two rows $\g_1,\g_2$ of $G_1$ satisfy $\g_1\w \g_2 = \bm{0} \in\G_1$. For example, for $n=8$,
	\[
	G_1 = \left[ \begin{array}{cccc|cccc}
		1 & 0 & 0 & 0 & 1 & 0 & 0 & 0  \\
		0 & 1 & 0 & 0 & 0 & 1 & 0 & 0  \\
		0 & 0 & 1 & 0 & 0 & 0 & 1 & 0  \\
		0 & 0 & 0 & 1 & 0 & 0 & 0 & 1  \\
	\end{array} \right] .
	\]
	These examples also show that for any positive integer $k$, there always exists a self-dual code $\C$ of parameters $[2k,k]$ which is also triorthogonal.
\end{ex}
\begin{ex} \label{not tri}
	Consider a linear self-dual code $\C$ with generator matrix
	\[
	G_2 = \begin{bmatrix}
		&1&0&0&0&0&1&1&1&\\
		&0&1&0&0&1&0&1&1&\\
		&0&0&1&0&1&1&0&1&\\
		&0&0&0&1&1&1&1&0&\\
	\end{bmatrix}
	= \begin{bmatrix}
		&\g_1&\\
		&\g_2&\\
		&\g_3&\\
		&\g_4&\\
	\end{bmatrix}.
	\]
	Unfortunately, $G_2$ is not a triorthogonal matrix, for $|\g_1\wedge \g_2\wedge \g_3| \equiv 1\pmod{2}$ and $\g_1\w \g_2 = (0,0,0,0,0,0,1,1)\notin\C$. Moreover,  $\g_i\w \g_j\notin\C$ for any $1\leq i<j\leq 4$.
\end{ex}

These two examples inspire the following theorem.

\begin{thm}\label{self-triorthogonal}
	If $C$ is a binary self-dual code, and also a triorthogonal space of length $2k$, then $[I_k\mid A]$ is a generator matrix of $C$ if and only if each row and each column of $A$ has only one $1$, i.e., $A$ can be changed to $I_k$ by permuting the columns.
\end{thm}
\begin{proof}
	Example \ref{all tri} has shown that $[I_k\mid I_k]$ is a generator matrix of a self-dual code which is also triorthogonal.
	
	On the other hand, assuming that $\x,\y$ are any two distinct rows of $[I_k\mid A]$, then the elements in the first $k$ coordinates of $\x\w\y$ are all zeros. This means if the last $k$ coordinates of $\x\w\y$ are not zeros, then $\x\w\y\notin \C$ because the first $k$ coordinates of all nonzero codewords in $\C$ cannot be all $0$, which means $\x\w\y = \zero$. Let $\x$ be a row of $[I_k\mid A]$, since $C$ is self-dual, and the first $k$ coordinates of $\x$ has exactly one $1$, then the number of $1$ in the last $k$ coordinates of $\x$ is odd. Combining these two conditions, we have that each row of $[I_k\mid A]$ can only have two ones, and we show it. We denote $A=(a_{ij})_{k\times k}$. Since \\
	Condition 1: for any two distinct rows $\x,\y$ of $A$, we have $\x\w\y = \zero$;\\
	Condition 2: for each row $\x$ of $A$, we have ${\mbox{wt}}(\x)$ is odd.\\
	Without loss of generality, let $a_{(1,1)},\dots,a_{(1,s_1)}$ be $1$ and $a_{(1,s_1+1)},\dots,a_{(1,k)}$ be $0$, we have for all $2\leq i\leq k$,  $a_{(i,1)},\dots,a_{(i,s_1)}$ must be $0$ by Condition 1. Now let $a_{(2,s_1+1)},\dots,a_{(2,s_2)}$ be $1$ and $a_{(2,s_2+1)},\dots,a_{(2,k)}$ be $0$, we have for all $3\leq i\leq k$, $a_{(i,s_1+1)},\dots,a_{(i,s_2)}$ must be $0$ by Condition 1. If we continue this process, we can get a sequence $\{s_i\}_{i=1}^k$ which satisfies
	\[\begin{cases}
			&\sum_{i=1}^k s_i = k, \\
			&s_i \text{ are odd for all }1\leq i\leq k. \text{ (By Condition 2)}
	\end{cases}\]
	This shows that $s_i=1$ for all $1\leq i\leq k$, which means that each row of $A$ has only one $1$, and each column of $A$ has only one $1$ by Condition 1. These completes the proof.
\end{proof}

\begin{re}
	Theorem \ref{self-triorthogonal} shows that Example \ref{all tri} is the only case for the binary self-dual codes which are also triorthogonal spaces, in the sense of permutation equivalence.
\end{re}

\subsection{Longer triorthogonal codes from self-dual codes}
Let $G$ be a triorthogonal matrix of size $m$-by-$n$ whose first $k$ $(k>0)$ rows have odd weight and the remaining rows have even weights. $G_1$ consists of these $k$ rows, i.e.,
\[
G = \left[ \begin{array}{c}
	G_1 \\
	\hline
	G_0 \\
\end{array}\right].
\]
We show that $G_0$ cannot be a generator matrix of any self-dual code.
\begin{thm}
	$G$ is a binary triorthogonal matrix of size $m$-by-$2n$ that has $n$ rows with even weight, and the remaining $m-n$ $(m-n > 0)$ rows with odd weight. If $\G_0$ denotes the span of all even-weight rows of $G$, then $\G_0$ is not a self-dual code.
\end{thm}
\begin{proof}
	Let $\G_0$ be a self-dual code, which means that $\G_0 = \G_0^{\perp},$ implying that each codeword in $\G_0$ has even weight. If we choose any odd-weight row $\x$ of $G$, then $|\x\w\y|\equiv 0\pmod{2}$ for any even-weight rows $\y$ of $G$ since $G$ is triorthogonal. Therefore, $\x\in\G_0^{\perp}=\G_0$ and ${\mbox{wt}}(\x)$ is odd, which is a contradiction.
\end{proof}

We give a construction for triorthogonal matrices with odd-weight rows by using the generator matrices of self-dual codes.
Let $\C$ be a self-dual codes with generator matrix $G$ of size $n$-by-$2n$, and $\C$ is also a triorthogonal space. We choose any triorthogonal matrix $A$ of size $t$-by-$a$, where the first $k$ rows of $A$ have odd weights and these $k$ rows form the matrix $A_1$. We have the following matrix
\[
G_1 = \left[ \begin{array}{c|c}
	A &O \\
	\hline
	O &G \\
\end{array}\right],
\]
where $O$ is the all-zero matrix.
By Proposition \ref{direct sum 1}, $G_1$ is triorthogonal. The following conclusions are the relationship between the parameters of the triorthogonal codes corresponding to \red{these matrices}.
\begin{lm}
	\red{If $A,B,C$ are three binary linear code} where $B\subsetneqq A$, then \[(A\oplus C)\setminus(B\oplus C) = (A\setminus B)\oplus C. \]
\end{lm}
\begin{proof}\red{
	One can easily check $(A\setminus B)\oplus C\subseteq (A\oplus C)\setminus(B\oplus C)$ and $|(A\setminus B)\oplus C| = (|A|-|B|)|C| = |A||C|-|B||C| = (A\oplus C)\setminus(B\oplus C)$. }
\end{proof}

\begin{thm}\label{repeat col}
	Over $\F_2$, let $B_{n\times 2n}$ be a generator matrix of a self-dual code which is also triorthogonal. Let $A_{a\times b}$ be a triorthogonal matrix whose first $k$ $(k>0)$ rows have odd weight and the remaining rows have even weights. Let $A_1$ be the matrix consisting of the first $k$ rows of matrix
	\[
	A = \left[ \begin{array}{c}
		A_1 \\
		\hline
		A_0
	\end{array} \right].
	\]
	If the triorthogonal code $\TriCode(A)$ has parameters $[[b,k,d_Z]]$, then the triorthogonal code $\TriCode(G)$ has parameters $[[2n+b,k,d_Z]]$ if
	\[
	G = \left[\begin{array}{c|c}
		A &O \\
		\hline
		O &B \\
	\end{array} \right]
	=\left[\begin{array}{c|c}
		A_1 &O \\
		\hline
		A_0 &O \\
		O &B \\
	\end{array} \right].
	\]
\end{thm}
\begin{proof}
	Let $\TriCode(G) = \CSS(\G_0;\G)$, and $\G_b$ be the span of all rows of $B$. By Lemma \ref{direct sum}, we have $\G = \G_a\oplus\G_b$, where $\G_a$ is the span of all the rows of $A$. Note $\G_0 = \G_{a}^0\oplus \G_b$, where $\G_{a}^0$ is the span of all the rows of $A_0$. By Lemma \ref{oplus and perp}, we have $\G_0^\perp = (\G_{a}^0)^\perp\oplus \G_b^\perp = (\G_{a}^0)^\perp\oplus \G_b$. Therefore,
	\[\begin{aligned}
		{\mbox{wt}}(\G_0^\perp \setminus \G^{\perp}) &= {\mbox{wt}}(((\G_{a}^0)^\perp\oplus \G_b)\setminus(\G_a^\perp\oplus\G_b)) \\
											&= {\mbox{wt}}(((\G_{a}^0)^\perp\setminus \G_a^\perp)  \oplus \G_b ) \\
											&= {\mbox{wt}}((\G_{a}^0)^\perp\setminus \G_a^\perp).
	\end{aligned}\]
	This completes the proof.
\end{proof}
\begin{re}
	By Example \ref{all tri}, for any positive integer $t$, if there exists a $[[n,k,d_Z]]$ triorthogonal code, then there exists an $[[n+2t,k,d_Z]]$ triorthogonal code.
\end{re}
\begin{ex}
	Here is a triorthogonal matrix
	\[
	A = \left[\begin{array}{ccccccccccccccc}
		0 &0 &0 &0 &1 &1 &1 &1 &1 &1 &0 &0 &0 &0 &1 \\
		1 &0 &0 &0 &1 &1 &1 &0 &0 &0 &0 &1 &1 &1 &1 \\
		0 &1 &0 &0 &1 &0 &0 &1 &1 &0 &1 &0 &1 &1 &1 \\
		0 &0 &1 &0 &0 &1 &0 &1 &0 &1 &1 &1 &0 &1 &1 \\
		0 &0 &0 &1 &0 &0 &1 &0 &1 &1 &1 &1 &1 &0 &1 \\
	\end{array} \right]_{5\times 15},
	\]
	and the distance of the matrix $A$ is $3$. Let $B = [\;I_3\mid I_3\;]$, where $I_3$ is the identity matrix of size $3$-by-$3$. Then the distance of the matrix $G$ is also $3$, where
	\[
	G =\left[  \begin{array}{c|c}
		A &O \\
		\hline
		O &B \\
	\end{array}\right].
	\]
\end{ex}

\subsection{Largest triorthogonal subcodes from self-dual codes}
Self-dual codes are a special kind of orthogonal spaces. In this section, we will give an algorithm to find the largest dimension of triorthogonal spaces that must exist as the subspaces of a self-dual code. In fact, this dimension will partly depend on the length of the self-dual codes. First, we give some required conclusions. Here we define
\[
\C_1 \w \C_2 = \{\x\w\y : \x\in\C_1,\y\in\C_2\}
\]
for any two binary linear codes $\C_1,\C_2\subseteq \F_2^n$.

\begin{lm}
	Let $\C$ be a binary linear code. If $H$ is a linearly independent set of $\C$, and $\x\w\y\in \C_1$ for any $\x,\y\in H$, where $\C_1$ is a linear subspace. Then $\H\w\H \subseteq \C_1$, where $\H$ is the span of $H$.
\end{lm}
\begin{proof}
	Similar to the proof of Lemma \ref{tri basis is tri}, so we omit it.
\end{proof}	

\begin{lm}
	Let $\C$ be a binary self-orthogonal code. If $H$ is a linearly independent set of $\C$ and satisfies that for any $\x,\y\in H$,
	$\x\w \y\in\rowspan\{H\}^{\perp}=\H^{\perp}$, then choose any $r$ elements in $\H$, we can get a triorthogonal matrix $G$ whose rows are these $r$ elements with $\rank(G)\leq min\{|H|,r\}$.
\end{lm}
\begin{proof}
	Since $\H\subseteq\C$, we have $|\x\w \y|\equiv0\pmod{2}$ for any $\x,\y\in\H$. For any $\x,\y,\z\in\H$, since $\x\w \y\in\H^{\perp}$, we have $|\x\w \y\w \z|\equiv0\pmod{2}$.
\end{proof}	

\begin{lm} \label{triorthogonal space generate form basis}
	If $\{\af_1,\dots,\af_s\}$ is a basis of a binary linear code $\C\subseteq \F_2^n$ and satisfies
	\[
	|\af_i\wedge\af_j\wedge\af_k| \equiv 0\pmod{2} \text{ for all } 1\leq i\leq j\leq k\leq s.
	\]
	Then for any $\x,\y,\z\in\C$, we have $|\x\wedge \y\wedge \z|\equiv0\pmod{2}$, i.e., $\C$ is a triorthogonal space. Also, if a codeword $\wf\in\F_2^n$ satisfies $|\wf\wedge\af_i\wedge \af_j| \equiv 0\pmod{2}$ for all $1\leq i\le j\leq s$, then $|\wf\wedge \y\wedge \z|\equiv0\pmod{2}$ for all $\y,\z\in\C$.
\end{lm}
\begin{proof}
	Let $\x,\y,\z\in\C$ and $\x = \sum_{i\in\Lambda_x}\a_i,\;\y = \sum_{j\in\Lambda_y}\a_j,\;\z = \sum_{k\in\Lambda_z}\a_k$, for some index sets $\Lambda_x, \Lambda_y$ and $\Lambda_z$. Then we have
	\[
	\begin{aligned}
		|\x\w \y\w \z|
		&= \left| (\sum_{i\in\Lambda_x}\a_i)\w(\sum_{j\in\Lambda_y}\a_j)\w(\sum_{k\in\Lambda_z}\a_k)\right|  \\
		&\equiv \sum_{i\in\Lambda_x}\sum_{j\in\Lambda_y}\sum_{k\in\Lambda_z}|\a_i\w \a_j\w \a_k| \pmod{2}\\
		&\equiv 0 \pmod{2},\text{ and}\\
		|\wf\w \y\w \z|
		&= \left| \wf\w(\sum_{j\in\Lambda_y}\a_j)\w(\sum_{k\in\Lambda_z}\a_k)\right| \\
		&\equiv \sum_{j\in\Lambda_y}\sum_{k\in\Lambda_z}|\wf\w \a_j\w \a_k| \pmod{2}\\
		&\equiv 0 \pmod{2}.
	\end{aligned}
	\]
	The result follows.
\end{proof}

It is clear that the subspaces of triorthogonal spaces are also triorthogonal, so we just need to find the largest triorthogonal spaces. Now we give a search algorithm to find the maximum triorthogonal space (i.e., the highest dimension) of a certain self-dual code.

\begin{al} \label{algo}
	Find the largest triorthogonal space of a binary self-dual code.
	\begin{center}
		\begin{tabularx}{\textwidth}{|lX|}
			\hline
			\multirow{2}{*}{\textbf{Input:}} &(1) A binary self-dual code $\C$ of parameters $[n,k]$ $(n=2k)$. \\
			&(2) Two non-zero and unequal codewords $\y,\z\in\C$ ($\y,\z$ are starting codewords). \\
			\hline
			\textbf{Step 1:} &Let $H = \{\y,\z\}$ and $\H = \rowspan\{H\}$. \\
			\textbf{Step 2:} &Solve the following system
			\begin{equation}\label{linear system 1}
				\begin{cases}
					|\x\wedge \af \wedge \be| \equiv 0 \pmod{2}, \text{ for any }\af,\be\in H\text{ and } \af\neq \be \\
					\x \in \C\\ 
					\x \notin \H
				\end{cases}.\end{equation}
			If there is no solution, then terminate this algorithm. Else, choose a solution $\x$ and let $\wf=\x$.\\
			\textbf{Step 3:} & Add $\wf$ to the set $H$. Update $\H = \rowspan\{H\}$ and turn to Step 2.\\
			\hline
			\textbf{Output:} & A linearly independent set $H\subseteq\C$ (We call $H$ the output set). \\
			\hline
		\end{tabularx}
	\end{center}
\end{al}

\begin{thm} \label{output thm}
	Let the output of Algorithm \ref{algo} be $H$, then $\H = \rowspan\{H\}$ is a triorthogonal space. Therefore, by choosing any $r$ elements in $H$, we can get a full-rank triorthogonal matrix of size $r$-by-$n$, which is also the generator matrix of a triorthogonal space.
\end{thm}
\begin{proof}
	By Algorithm \ref{algo}, since $\C$ is self-dual (i.e., $\C=\C^\perp$), $H$ satisfies $|\af\w\be\w\rf| \equiv 0\pmod{2}$ for any $\af,\be,\rf\in H$. By Lemma \ref{triorthogonal space generate form basis}, $\H$ is triorthogonal. Since each subspace of a triorthogonal space is also triorthogonal, we obtain the rest of this theorem.
\end{proof}

Now we can gain a lot of triorthogonal spaces with parameters $[n,r]$ $(r\leq |H|)$ which are the subspaces of $\H$.
Now we give some bounds of $|H|$. The following conclusions follow the notations in Algorithm \ref{algo}.

\begin{thm}\label{bound}
	Let the output of Algorithm \ref{algo} be $H$, then
	\[
	\left\lceil \frac{\sqrt{8k+1}-1}{2} \right\rceil\leq |H|\leq k .
	\]
	Also $|H| = k$ if and only if $\C$ is a triorthogonal space. Furthermore, Algorithm \ref{algo} will continue to proceed until the conditions for terminating the algorithm are met, i.e., \eqref{linear system 1} has no solution.
\end{thm}
\begin{proof}
	Note that
	\begin{equation}\label{linear system 2}
		\begin{cases}
			|\x\wedge \af \wedge \be| \equiv 0 \pmod{2}, \text{ for any }\af,\be\in H\text{ and } \af\neq \be, \\
			\x \in \C. 
		\end{cases}.\end{equation}
	Since $\C$ is a binary self-dual code, we have
	\[
	\x\in\C  \Leftrightarrow |\x\w \bm{\gamma}_i|\equiv 0\pmod{2} \text{ for all } 1\leq i\leq k,
	\]
	where $\{\bm{\gamma}_1,\cdots,\bm{\gamma}_k\}$ is a basis of $\C$. Therefore, \eqref{linear system 2} is actually a homogeneous system of linear equations which has $n$ unknowns and $k+{|H|\choose 2}$ different equations over $\F_2$. From the knowledge of linear algebra, the dimension of the solution space $\X$ of the system \eqref{linear system 2} satisfies $\dim(\X)\geq n-(k+{|H|\choose 2})$, as these $k+{|H|\choose 2}$ different linear equations may have several linearly dependent equations. Since $\C$ is a self-dual code, we have $n = 2k$ and \[\dim(\X)\geq k-{|H|\choose 2}.\]
	Furthermore, $|\H| = 2^{|H|}$ and each codeword in $\H$ satisfies the system \eqref{linear system 2}, so if the system \eqref{linear system 1} has a solution, then we need $\dim(\X)> \dim(\H)$. However, since $H$ is the output, which means that \eqref{linear system 1} has no solution, then we have
	$\dim(\X) \leq \dim(\H)$. Therefore, $k \leq |H| + {|H|\choose 2} = {|H|+1\choose 2}$ and \[ |H|\geq\left\lceil \frac{\sqrt{8k+1}-1}{2} \right\rceil. \]
	
	What's more, if $\C$ is a triorthogonal space, then $|H| = k$. Conversely, if $|H|=k$, since $H\subseteq \C$, then $\H = \C$ and $\C$ is triorthogonal space by Theorem \ref{output thm}.
	
	From our discussion above, to prove that the algorithm can continue until the termination condition is reached, we only need to prove that $f(|H|) = \dim(\X) - \dim(\H)$ is a strictly monotonically decreasing function, where $\H=\rowspan\{H\}$ and $H$ is the set before the algorithm ends, not the output. During the process of the algorithm, when $|H|$ keeps increasing, since the number of unknowns in \eqref{linear system 2} remains unchanged (always $n$), \red{and the number of equations increases gradually. These facts illustrate that $\dim(\X)$ will become smaller or remain the same, while $\dim(\H) = |H|$ will become larger. This fact shows that $f$ is a strictly monotonically decreasing function.} The result follows.
\end{proof}

\begin{thm} \label{has solutions 1}
	\red{During the process of the algorithm, if} $k > {|H|+1\choose 2}$, then \eqref{linear system 1} always has solutions.
\end{thm}
\begin{proof}
	Following the notations used in Theorem \ref{bound} and their meanings, if $k > {|H|+1\choose 2}$, then
	\[
	\dim(\X)\geq k-{|H|\choose 2} > |H| = \dim(\H),
	\]
	and \eqref{linear system 1} has solutions.
\end{proof}

We denote $\one_n$ as the all-one vector of length $n$. Since for each codeword $\x$ in a self-dual code $\C\subseteq \F_2^n$, $|\x\w\one_n| = |\x| \equiv 0 \pmod{2}$, so $\one_n\in\C^\perp = \C$. We choose $\one_n$ as one of the starting codewords. Then \eqref{linear system 2} actually has $k + {|H|-1\choose 2}$ different linear equations where $H$ is the set during the execution of the algorithm and $\H = \rowspan\{H\}$, as for any $\af\in H$, $|\x\w\one_n\w\af| = |\x\w\af|$, and $\x\in\C$, so $|\x\w\af|\equiv 0\pmod{2}$. Therefore, \eqref{linear system 2} is actually
\[
\begin{cases}
	|\x\wedge \af \wedge \be| \equiv 0 \pmod{2}, \text{ for any }\af,\be\in H\setminus\{\one_n\}\text{ and } \af\neq \be,  \\
	\x \in \C. 
\end{cases}
\]
If $\dim(\X)\leq \dim(\H)$, i.e., \eqref{linear system 1} has no solution, then
\[\begin{aligned}
	&n-(k + {|H|-1\choose 2}) = k - {|H|-1\choose 2} \leq \dim(\X)\leq \dim(\H) = |H| \\
	&\Leftrightarrow         |H| +  {|H|-1\choose 2} = |H| + {|H|\choose 2} - (|H|-1) = {|H|\choose 2}+1 \geq k \\
	&\Leftrightarrow   |H| \geq \frac{\sqrt{8k-7}+1}{2}.
\end{aligned}\]

\begin{thm}\label{impro bound}
	Let the output of Algorithm \ref{algo} be $H$. Then
	\[
	|H| \geq \left\lceil \frac{\sqrt{8k-7}+1}{2} \right\rceil
	\]
	if $\one_n$ is one of the starting codewords and $\H=\rowspan\{H\}$ is a unital triorthogonal space.
\end{thm}

\begin{thm}\label{has solutions 2}
	Let $\one_n$ be one of the starting codewords. During the process of the algorithm, if $k > {|H|\choose 2}+1$, then \eqref{linear system 1} always has solutions.
\end{thm}

\begin{ex}\label{k=4 self-dual}
	Consider a self-dual code $\C_2$ with generator matrix
	\[
	G_2 = \left[ \begin{array}{cccccccccc}
		1&0&0&0&1&0&1&0&0&1 \\
		0&1&0&0&1&0&1&0&1&0 \\
		0&0&1&0&0&0&1&0&1&1 \\
		0&0&0&1&1&0&0&0&1&1 \\
		0&0&0&0&0&1&0&1&0&0 \\
	\end{array}\right]. \]
	By using Algorithm \ref{algo}, and choosing $\one_{10},(1,0,0,0,1,0,1,0,0,1)$ as the starting codewords, we can get the output set
	\[\begin{aligned}
		H = \{ &(1,0,0,0,1,0,1,0,0,1),(0,1,0,0,1,0,1,0,1,0),\\
			   &(0,0,0,0,0,1,0,1,0,0),(1,1,1,1,1,1,1,1,1,1) \}.
	\end{aligned}\]
\end{ex}

\red{We focus on such a problem, i.e., the existence of $[2k,r]$ triorthogonal spaces which are the subspaces of a certain self-dual code $\C$ of parameters $[2k,k]$.} During the operation of the algorithm, when $|H| = r-1$, if \eqref{linear system 1} has solutions, then the number of the elements in the output set is greater than or equal to $r$, which also means that there exists a triorthogonal space of dimension $r$. Therefore, combining Theorem \ref{has solutions 1} and Theorem \ref{has solutions 2}, we have the following theorem.

\begin{thm} \label{conclusion of require dimension}
	Follow the notations and their meanings in Algorithm \ref{algo}. For a positive integer $r$ $(r\geq 3)$, if $k> {r\choose 2}$, then there exists a $[n,r]$ triorthogonal space, which is the subspace of the self-dual code $\C$. Moreover, if $\one_n$ is one of starting codewords and $k >{r-1\choose 2}+1$, we have the same conclusion.
\end{thm}

We put some numerical results in Table \ref{min k table}.
\begin{table}[htb]
	\caption{For some fixed $r$, the minimal value of $k$ in Theorem \ref{conclusion of require dimension}, which satisfies that there exists a $[2k,r]$ $(r\leq k)$ triorthogonal subspace which is the subspace of \red{a self-dual $[2k,k]$ code} $\C$.}
	\label{min k table}
	\begin{center}
		\begin{tabular}{m{3.5cm}<{\centering}|m{3.5cm}<{\centering}|m{3.5cm}<{\centering}}
			\hline
			
			\hline
			\# & General case & $\one_{2k}$ is a starting codeword \\
			\hline
			$r$        &$k = {r\choose 2}+1$       & $k = {r-1\choose 2} + 2$  \\
			\hline
			3  & 4  & 3 \\
			4  & 7  & 5 \\
			5  & 11  & 8 \\
			6  & 16  & 12 \\
			7  & 22  & 17 \\
			8  & 29  & 23 \\
			9  & 37  & 30 \\
			10  & 46  & 38 \\
			\hline
			
			\hline
		\end{tabular}
	\end{center}
\end{table}

\section{Applications and examples}

%
%

\rR{
In \cite{data paper}, the authors classified triorthogonal codes with small parameters and gave the distances of triorthogonal matrices corresponding to triorthogonal codes for each pair of $[[n,k]]$, where $n+k\le 38$. This section makes use of the methods mentioned earlier to construct the triorthogonal matrices and triorthogonal codes. First, we give a example to show that the codes in \cite{data paper} can be obtained by using our methods.}
\begin{ex}
	Here is a full-rank triorthogonal matrix of size $5$-by-$16$:
	\[
	G = \left[\begin{array}{cccccccccccccccc}
		1 &1 &1 &1 &1 &1 &1 &1 &1 &1 &1 &1 &1 &1 &1 &1 \\
		0 &1 &0 &0 &0 &1 &1 &1 &0 &0 &0 &0 &1 &1 &1 &1 \\
		0 &0 &1 &0 &0 &1 &0 &0 &1 &1 &0 &1 &0 &1 &1 &1 \\
		0 &0 &0 &1 &0 &0 &1 &0 &1 &0 &1 &1 &1 &0 &1 &1 \\
		0 &0 &0 &0 &1 &0 &0 &1 &0 &1 &1 &1 &1 &1 &0 &1 \\
	\end{array} \right]_{5\times 16}.
	\]
	After adding the $2$, $3$, $4$, \red{$5$-th rows of $G$} to the first row, we have
	\[G_1 = \left[\begin{array}{cccccccccccccccc}
		1 &0 &0 &0 &0 &1 &1 &1 &1 &1 &1 &0 &0 &0 &0 &1 \\
		0 &1 &0 &0 &0 &1 &1 &1 &0 &0 &0 &0 &1 &1 &1 &1 \\
		0 &0 &1 &0 &0 &1 &0 &0 &1 &1 &0 &1 &0 &1 &1 &1 \\
		0 &0 &0 &1 &0 &0 &1 &0 &1 &0 &1 &1 &1 &0 &1 &1 \\
		0 &0 &0 &0 &1 &0 &0 &1 &0 &1 &1 &1 &1 &1 &0 &1 \\
	\end{array} \right]_{5\times 16}.\]
	Deleting the first column of $G_1$, we have
	\[
	G_2 = \left[\begin{array}{ccccccccccccccc}
		0 &0 &0 &0 &1 &1 &1 &1 &1 &1 &0 &0 &0 &0 &1 \\
		1 &0 &0 &0 &1 &1 &1 &0 &0 &0 &0 &1 &1 &1 &1 \\
		0 &1 &0 &0 &1 &0 &0 &1 &1 &0 &1 &0 &1 &1 &1 \\
		0 &0 &1 &0 &0 &1 &0 &1 &0 &1 &1 &1 &0 &1 &1 \\
		0 &0 &0 &1 &0 &0 &1 &0 &1 &1 &1 &1 &1 &0 &1 \\
	\end{array} \right]_{5\times 15}.
	\]
	$\TriCode(G_2)$ has parameters $[[15,1,3]]$. Moreover, if we delete the first column of $G_2$, then the new matrix
	\[
	G_3 = \left[\begin{array}{cccccccccccccc}
		0 &0 &0 &1 &1 &1 &1 &1 &1 &0 &0 &0 &0 &1 \\
		0 &0 &0 &1 &1 &1 &0 &0 &0 &0 &1 &1 &1 &1 \\
		1 &0 &0 &1 &0 &0 &1 &1 &0 &1 &0 &1 &1 &1 \\
		0 &1 &0 &0 &1 &0 &1 &0 &1 &1 &1 &0 &1 &1 \\
		0 &0 &1 &0 &0 &1 &0 &1 &1 &1 &1 &1 &0 &1 \\
	\end{array} \right]_{5\times 14}
	\]
	generates $\TriCode(G_3)$ of parameters $[[14,2,2]]$. By Theorem \ref{dirct sum of tricodes}, then
	\[
	G_4 = \left[ \begin{array}{c|c}
		G_2   &O \\
		\hline
		O     &G_3 \\
	\end{array} \right] \text{ and }
	G_5 = \left[ \begin{array}{c|c}
		G_2   &O \\
		\hline
		O     &G_2 \\
	\end{array} \right]
	\]
	can generate $\TriCode(G_4)$ of parameters $[[29,3,2]]$ and $\TriCode(G_5)$ of parameters $[[30,2,3]]$, respectively. One can also check these results by Magma. \rR{This example illustrates that part of triorthogonal codes in \cite{data paper} can also be constructed by using our methods.}
\end{ex}

\rR{
	\cite{data paper} has given many triorthogonal codes with parameters $[[n,k,d_Z]]$, where $n+k\le 38$ and each code in \cite{data paper} can be constructed as a descendant of some unital triorthogonal subspaces. By using these codes, many new triorthogonal codes with parameters $[[n,k,d_Z]]$ can be obtained and we list them in Table \ref{examples} , where $n+k\ge 40$ and the value of $d_Z$ is the best that our methods can achieve.
	Table \ref{examples} can be seen as a supplementary of Table II in \cite{data paper}. It is worth noting that Table \ref{examples} can actually continue to be supplemented by using the methods in this paper.}

\begin{table}[]
	\caption{\rR{The triorthogonal codes with parameters $[[n,k,d_Z]]$ that we obtain, where $n+k\ge 40$ and  the value of $d_Z$ is the best that our methods can achieve.}}
	\label{examples}
	\begin{center}
		\setlength{\tabcolsep}{6mm}{
			\begin{tabular}{|c|cccccc|}
				\hline
				\multirow{2}{*}{$n$} & \multicolumn{6}{c|}{$d_Z$}                                                                                                                             \\ \cline{2-7}
				& \multicolumn{1}{c|}{$k=2$} & \multicolumn{1}{c|}{$k=3$} & \multicolumn{1}{c|}{$k=4$} & \multicolumn{1}{c|}{$k=5$} & \multicolumn{1}{c|}{$k=6$} & $k=7$ \\ \hline
				38                   & \multicolumn{1}{c|}{3}     & \multicolumn{1}{c|}{--}    & \multicolumn{1}{c|}{--}    & \multicolumn{1}{c|}{--}    & \multicolumn{1}{c|}{2}     & --    \\ \hline
				39                   & \multicolumn{1}{c|}{--}    & \multicolumn{1}{c|}{--}    & \multicolumn{1}{c|}{--}    & \multicolumn{1}{c|}{2}     & \multicolumn{1}{c|}{--}    & 1     \\ \hline
				40                   & \multicolumn{1}{c|}{3}     & \multicolumn{1}{c|}{--}    & \multicolumn{1}{c|}{2}     & \multicolumn{1}{c|}{--}    & \multicolumn{1}{c|}{2}     & --    \\ \hline
				41                   & \multicolumn{1}{c|}{--}    & \multicolumn{1}{c|}{2}     & \multicolumn{1}{c|}{--}    & \multicolumn{1}{c|}{2}     & \multicolumn{1}{c|}{--}    & 2     \\ \hline
				42                   & \multicolumn{1}{c|}{3}     & \multicolumn{1}{c|}{--}    & \multicolumn{1}{c|}{2}     & \multicolumn{1}{c|}{--}    & \multicolumn{1}{c|}{2}     & --    \\ \hline
				43                   & \multicolumn{1}{c|}{--}    & \multicolumn{1}{c|}{3}     & \multicolumn{1}{c|}{--}    & \multicolumn{1}{c|}{2}     & \multicolumn{1}{c|}{--}    & 2     \\ \hline
				44                   & \multicolumn{1}{c|}{3}     & \multicolumn{1}{c|}{--}    & \multicolumn{1}{c|}{2}     & \multicolumn{1}{c|}{--}    & \multicolumn{1}{c|}{2}     & --    \\ \hline
				45                   & \multicolumn{1}{c|}{--}    & \multicolumn{1}{c|}{3}     & \multicolumn{1}{c|}{--}    & \multicolumn{1}{c|}{2}     & \multicolumn{1}{c|}{--}    & 2     \\ \hline
				46                   & \multicolumn{1}{c|}{3}     & \multicolumn{1}{c|}{--}    & \multicolumn{1}{c|}{2}     & \multicolumn{1}{c|}{--}    & \multicolumn{1}{c|}{2}     & --    \\ \hline
				47                   & \multicolumn{1}{c|}{--}    & \multicolumn{1}{c|}{3}     & \multicolumn{1}{c|}{--}    & \multicolumn{1}{c|}{2}     & \multicolumn{1}{c|}{--}    & 2     \\ \hline
				48                   & \multicolumn{1}{c|}{3}     & \multicolumn{1}{c|}{--}    & \multicolumn{1}{c|}{2}     & \multicolumn{1}{c|}{--}    & \multicolumn{1}{c|}{2}     & --    \\ \hline
				49                   & \multicolumn{1}{c|}{--}    & \multicolumn{1}{c|}{3}     & \multicolumn{1}{c|}{--}    & \multicolumn{1}{c|}{2}     & \multicolumn{1}{c|}{--}    & 2     \\ \hline
				50                   & \multicolumn{1}{c|}{3}     & \multicolumn{1}{c|}{--}    & \multicolumn{1}{c|}{2}     & \multicolumn{1}{c|}{--}    & \multicolumn{1}{c|}{2}     & --    \\ \hline
				51                   & \multicolumn{1}{c|}{--}    & \multicolumn{1}{c|}{3}     & \multicolumn{1}{c|}{--}    & \multicolumn{1}{c|}{2}     & \multicolumn{1}{c|}{--}    & 2     \\ \hline
				52                   & \multicolumn{1}{c|}{3}     & \multicolumn{1}{c|}{--}    & \multicolumn{1}{c|}{2}     & \multicolumn{1}{c|}{--}    & \multicolumn{1}{c|}{2}     & --    \\ \hline
				53                   & \multicolumn{1}{c|}{--}    & \multicolumn{1}{c|}{3}     & \multicolumn{1}{c|}{--}    & \multicolumn{1}{c|}{2}     & \multicolumn{1}{c|}{--}    & 2     \\ \hline
				54                   & \multicolumn{1}{c|}{3}     & \multicolumn{1}{c|}{--}    & \multicolumn{1}{c|}{2}     & \multicolumn{1}{c|}{--}    & \multicolumn{1}{c|}{2}     & --    \\ \hline
				55                   & \multicolumn{1}{c|}{--}    & \multicolumn{1}{c|}{3}     & \multicolumn{1}{c|}{--}    & \multicolumn{1}{c|}{2}     & \multicolumn{1}{c|}{--}    & 2     \\ \hline
				56                   & \multicolumn{1}{c|}{3}     & \multicolumn{1}{c|}{--}    & \multicolumn{1}{c|}{3}     & \multicolumn{1}{c|}{--}    & \multicolumn{1}{c|}{2}     & --    \\ \hline
				57                   & \multicolumn{1}{c|}{--}    & \multicolumn{1}{c|}{3}     & \multicolumn{1}{c|}{--}    & \multicolumn{1}{c|}{2}     & \multicolumn{1}{c|}{--}    & 2     \\ \hline
				58                   & \multicolumn{1}{c|}{3}     & \multicolumn{1}{c|}{--}    & \multicolumn{1}{c|}{3}     & \multicolumn{1}{c|}{--}    & \multicolumn{1}{c|}{2}     & --    \\ \hline
				59                   & \multicolumn{1}{c|}{--}    & \multicolumn{1}{c|}{3}     & \multicolumn{1}{c|}{--}    & \multicolumn{1}{c|}{2}     & \multicolumn{1}{c|}{--}    & 2     \\ \hline
				60                   & \multicolumn{1}{c|}{3}     & \multicolumn{1}{c|}{--}    & \multicolumn{1}{c|}{3}     & \multicolumn{1}{c|}{--}    & \multicolumn{1}{c|}{2}     & --    \\ \hline
				61                   & \multicolumn{1}{c|}{--}    & \multicolumn{1}{c|}{3}     & \multicolumn{1}{c|}{--}    & \multicolumn{1}{c|}{2}     & \multicolumn{1}{c|}{--}    & 2     \\ \hline
				62                   & \multicolumn{1}{c|}{3}     & \multicolumn{1}{c|}{--}    & \multicolumn{1}{c|}{3}     & \multicolumn{1}{c|}{--}    & \multicolumn{1}{c|}{2}     & --    \\ \hline
				63                   & \multicolumn{1}{c|}{--}    & \multicolumn{1}{c|}{3}     & \multicolumn{1}{c|}{--}    & \multicolumn{1}{c|}{3}     & \multicolumn{1}{c|}{--}    & 2     \\ \hline
				64                   & \multicolumn{1}{c|}{3}     & \multicolumn{1}{c|}{--}    & \multicolumn{1}{c|}{3}     & \multicolumn{1}{c|}{--}    & \multicolumn{1}{c|}{2}     & --    \\ \hline
				65                   & \multicolumn{1}{c|}{--}    & \multicolumn{1}{c|}{3}     & \multicolumn{1}{c|}{--}    & \multicolumn{1}{c|}{3}     & \multicolumn{1}{c|}{--}    & 2     \\ \hline
				66                   & \multicolumn{1}{c|}{3}     & \multicolumn{1}{c|}{--}    & \multicolumn{1}{c|}{3}     & \multicolumn{1}{c|}{--}    & \multicolumn{1}{c|}{2}     & --    \\ \hline
			\end{tabular}
		}
	\end{center}
\end{table}

\section{Conclusion}

Our paper is divided into two main parts.
The goal of the first part is to construct new triorthogonal matrices from known ones. Shortening, extending, direct sum, and some other classical techniques have been used, and the parameters of the corresponding quantum codes have been also discussed.

The content of the second part is to find the triorthogonal matrices by using the self-dual codes.

\begin{itemize}
\item When is a self-dual code also triorthogonal? A necessary and sufficient condition is that the code is, up to equivalence, a direct sum of repetitions codes of length two.

\item Regarding the question of the existence of triorthogonal spaces that are the subspaces of self-dual codes, we give \red{an algorithm} to find such triorthogonal spaces. Also, by this algorithm, we give some lower bounds of the dimension of such triorthogonal spaces.
\end{itemize}

\rR{
In the last section of our paper, by using the triorthogonal codes of table II in \cite{data paper}, we continue to search for many triorthogonal codes with parameters $[[n,k,d_Z]]$, where $n+k\ge 40$. We have included these results in Table \ref{examples}, which can be seen as a supplementary of Table II in \cite{data paper}
}

It will be an interesting problem to construct an infinite family of triorthogonal matrices or triorthogonal codes  from cyclic codes, 2-designs, or strongly regular graphs.



\section{Conflict of Interest}

The authors have no conflicts of interest to declare that are relevant to the content
of this article.

\section{Data Deposition Information}

Our data can be obtained from the authors upon reasonable request.

\end{document}